\documentclass[runningheads]{llncs}

\usepackage{graphicx}
\usepackage[sort]{cite}
\usepackage{enumerate}
\usepackage{booktabs}
\usepackage{tabu}
\usepackage[ruled,lined,boxed,linesnumbered]{algorithm2e}
\usepackage{amsmath,amssymb}
\usepackage{todonotes}
\usepackage{algorithmic}

\newtheorem{Reduction Rule}{Reduction Rule}

\usepackage{thmtools}
\usepackage{thm-restate}

\usepackage{hyperref}

\usepackage{cleveref}

\usepackage{microtype}
\usepackage{xspace}
\usepackage{color}

\newcommand{\stpath}{{$s$-$t$ path}\xspace}
\newcommand{\stgraph}{{$s$-$t$ graph}\xspace}
\newcommand{\stpaths}{{$s$-$t$ paths}\xspace}
\newcommand{\sstpaths}{{shortest $s$-$t$ paths}\xspace}

\newcommand{\vc}{{\sc Vertex Cover}\xspace}

\newcommand{\tp}{{\sc Tracking Paths}\xspace}
\newcommand{\tsp}{{\sc Tracking Shortest Paths}\xspace}

\newcommand{\fvs}{{\sc Feedback Vertex set}\xspace}

\newcommand{\NP}{\text{\normalfont  NP}\xspace}

\newcommand{\FPT}{\text{\normalfont FPT}\xspace}

\newcommand{\APX}{{\sc APX}\xspace}

\newcommand{\Oh}{\mathcal{O}}

 \newcommand{\defparproblem}[4]{
  \vspace{1mm}
\noindent\fbox{
  \begin{minipage}{0.96\textwidth}
  \begin{tabular*}{\textwidth}{@{\extracolsep{\fill}}lr} #1  & {\bf{Parameter:}} #3
\\ \end{tabular*}
  {\bf{Input:}} #2  \\
  {\bf{Question:}} #4
  \end{minipage}
  }
  \vspace{1mm}
}

\begin{document}

\title{Improved Kernels for Tracking Path Problem} 

\author{Pratibha Choudhary \inst{1} \and Venkatesh Raman \inst{2}}

\institute{Indian Institute of Technology Jodhpur, Jodhpur, India.\\
\email{pratibhac247@gmail.com}
\and
Institute of Mathematical Sciences, HBNI, Chennai, India.\\
\email{vraman@imsc.res.in}
}

\authorrunning{P. Choudhary and V. Raman}
\maketitle 

\begin{abstract}
Tracking of moving objects is crucial to security systems and networks. 
Given a graph $G$, terminal vertices $s$ and $t$, and an integer $k$, the \tp problem asks whether there exists at most $k$ vertices, which if marked as trackers, would ensure that the sequence of trackers encountered in each \stpath is unique. It is known that the problem is \textsc{NP}-hard and admits a kernel (reducible to an equivalent instance) with $\Oh(k^6)$ vertices and $\Oh(k^7)$ edges, when parameterized by the size of the output (tracking set) $k$~\cite{tr-j}. An interesting question that remains open is whether the existing kernel can be improved. In this paper we answer this affirmatively:
\begin{enumerate}[(i)]
\item For general graphs, we show the existence of a kernel of size $\Oh(k^2)$
\item For planar graphs, we improve this further by giving a kernel of size $\Oh(k)$ 
\end{enumerate}
In addition, we also show that finding a tracking set of size at most $n-k$ for a graph on $n$ vertices is hard for the parameterized complexity class {\sc W[1]}, when parameterized by $k$.
\keywords{Graphs \and Paths \and Kernelization \and Fixed-parameter tractability \and Planar graphs \and Tracking Paths \and Below guarantee.}
\end{abstract}

\section{Introduction}
\label{sec:intro}
Graphs serve as a systematic model for modeling and analysis of many real life problems. One of the commonly studied problems in areas of networks and machine learning is tracking of moving objects. Typically a secure environment or setup that needs to be monitored, has one or more source and destination points. The requirement is usually to identify the path traced by entities in a network. This can be implemented by placing trackers at some of the checkpoints.

Coordinated path tracking and framework for multi-target tracking have been discussed in~\cite{coordinated-tracking} and \cite{tracking-framework}. Tracking of moving objects has been widely studied in networks, wireless sensor networks, neural networks and binary sensor networks~\cite{tracking-moving-obj},\cite{tracking-moving-obj2},\cite{tracking-neural},\cite{tracking-binary}. Tracking algorithms can also be used in designing debugging tools in programs and for leakage detection systems. Resource efficient solutions are of key important in such scenarios.

The problem of target tracking can be modeled as the following graph theoretic problem.
Let $G=(V,E)$ be an undirected graph without any self loops or parallel edges with a unique entry vertex (source) $s$ and a unique exit vertex (destination) $t$. A simple path from $s$ to $t$ is called an \stpath.
The \tp problem asks to find a set of vertices $T\subseteq V$ such that for any two distinct \stpaths, say $P_1$ and $P_2$, the sequence of vertices in $T\cap V(P_1)$ as encountered in $P_1$ is different from the sequence of vertices in $T\cap V(P_2)$ as encountered in $P_2$. Here $T$ is called a \textit{tracking set} for the graph $G$, and the vertices in $T$ are called \textit{trackers}.
Banik et al.~\cite{ciac17} first studied the problem of tracking paths in graphs, where they focused on distinguishing all \sstpaths in a graph and proved that the problem (\tsp) \NP-hard and \APX-hard. They also gave a $2$-approximate algorithm for the same in planar graphs, along with giving some other results. \tsp was first studied from a parameterized perspective in~\cite{caldam18},~\cite{tr1-j}, where the problem was shown to be fixed-parameter tractable (\textsc{FPT}). \tp is formally defined as follows.

\defparproblem{\tp $(G,s,t,k)$}{An undirected  graph $G=(V,E)$ with two distinguished vertices $s$ and $t$, and a non-negative integer $k$.}
{$k$}
{Does there exist a tracking set  $T$ of size at most $k$ for $G$?}
\medskip

Banik et al.~\cite{tr-j} proved \tp to be \textsc{NP}-complete and fixed-parameter tractable by showing the existence of a polynomial kernel. Specifically it was proven that an instance of \tp can be reduced to an equivalent instance of size $\Oh(k^7)$ in polynomial time, where $k$ is the desired size of the tracking set\footnote{Throughout the paper we assume $k$ to be a non-negative integer.}. Here it remains an open question whether the kernel size can be improved and if there possibly exists a lower bound for the kernel. We answer the first question in this paper by giving an improved kernel for general graphs and planar graphs. We also give the first hardness result for \tp with respect to parameterized complexity.

\noindent 
{\bf Our Contributions and Methods.} We give a quadratic kernel for \tp on general graphs, which is a major improvement from the $\Oh(k^7)$ kernel given in~\cite{tr-j}. We also give a linear kernel for \tp on planar graphs. Further we prove that deciding if there exists a tracking set of size at most $n-k$, where $n$ is the number of vertices in the graph, is {\sc W[1]}-hard.

Given an instance $(G,s,t,k)$, we give a polynomial time algorithm that either determines that $(G,s,t,k)$ is a NO instance or produces an equivalent instance  with $\Oh(k^2)$ vertices and $\Oh(k^2)$ edges, where $k$ is the size of a desired tracking set. This polynomial time algorithm is called a {\em kernelization algorithm} and the reduced instance is called a {\em kernel}. For more details about parameterized complexity and kernelization we refer to monographs~\cite{DF99,Cygan:2015:PA:2815661}.

The kernelization algorithm works along the following lines. Let $(G,s,t,k)$ be an input instance to \tp. Two main results used to build the kernel in~\cite{tr-j} were $(i)$ every tracking set is also a \textit{feedback vertex set} (set of vertices whose deletion removes all cycles from graph), $(ii)$ if there exist more than $k+1$ paths between a pair of vertices in $G$, then $G$ cannot be tracked with at most $k$ trackers. However, in this paper, though we start the algorithm with a $2$-approximate solution for \textit{feedback vertex set} (FVS), we use another newly introduced result. Specifically, we prove that if there exists an induced subgraph in $G$ which consists of a tree with all of its leaves adjacent to a particular vertex $v$, then the size of a minimum tracking set for $G$ is at least one less than the number of neighbors of $v$ in this tree. 
Then if $S$ is an FVS of size at most $2k$, we give bounds on different types of vertices in $G\setminus S$, based on how they share neighbors in $S$. Combining all these bounds we prove the existence of a quadratic kernel for general graphs. We also give a linear kernel for planar graphs. A planar graph is a graph that can be embedded on a two dimensional plane i.e. it can be drawn on a two dimensional plane in such a way that its edges intersect only at their end points and do not cross each other. Planar graphs are also those graphs that do not have $K_{3,3}$ and $K_5$ as a minor. Eppstein et al. studied \tp for planar graphs in~\cite{ep-planar}, where they show that \tp remains \textsc{NP}-complete when the graph is planar, and give a $4$-approximation algorithm for this setting. Our linear kernel uses the bound on the number of faces with respect to the size of an optimal tracking set for a planar graph, given in~\cite{ep-planar}.

%
In Parameterized complexity it is common to identify tractable parameterizations. For a graph on $n$ vertices, $n$ trackers is a trivial upper bound, so from the perspective of `distance to triviality' we consider the question of whether we can find a tracking set with $n-k$ trackers. Many graph theory problems have contrasting results in similar cases. For example, in \textsc{Graph Coloring} (coloring the vertices such that end points of each edge are colored differently), finding if we can color the graph vertices with at most $k$ colors is \textit{para} \textsc{NP}-hard (\textsc{NP}-hard even if $k$ is a constant), but finding if we can color the graph vertices with at most $n-k$ colors has a \textsc{FPT} and admits a $O(k^2)$ kernel~\cite{Cygan:2015:PA:2815661}. Similarly, \textsc{Vertex Cover} admits a linear kernel if the question is whether there exists a vertex cover (set of vertices which include end points of all edges in the graph) of size at most $k$, while the problem of finding whether there exists a vertex cover of size $n-k$ is {\sc W[1]}-hard~\cite{Cygan:2015:PA:2815661}.
We prove that finding a tracking set of size at most $n-k$ is {\sc W[1]}-hard on a graph with $n$ vertices, where the parameter is $k$. 

Although a tracking set is also a feedback vertex set, both are fundamentally very different. A graph may have a small FVS but the tracking set may be arbitrarily larger than the FVS. Moreover, \tp is more demanding as a problem compared to the classic covering problems studied in graph theory. While covering problems aim at hitting a particular type of structure in graphs, \tp requires distinguishing each \stpath uniquely using a small set of vertices. Here we mention an important property about this problem. It can be shown using an example that the problem is not closed on minors. See Figure~\ref{fig:minor}.

\begin{figure}[ht]
    \centering
    \includegraphics[width=5.0in]{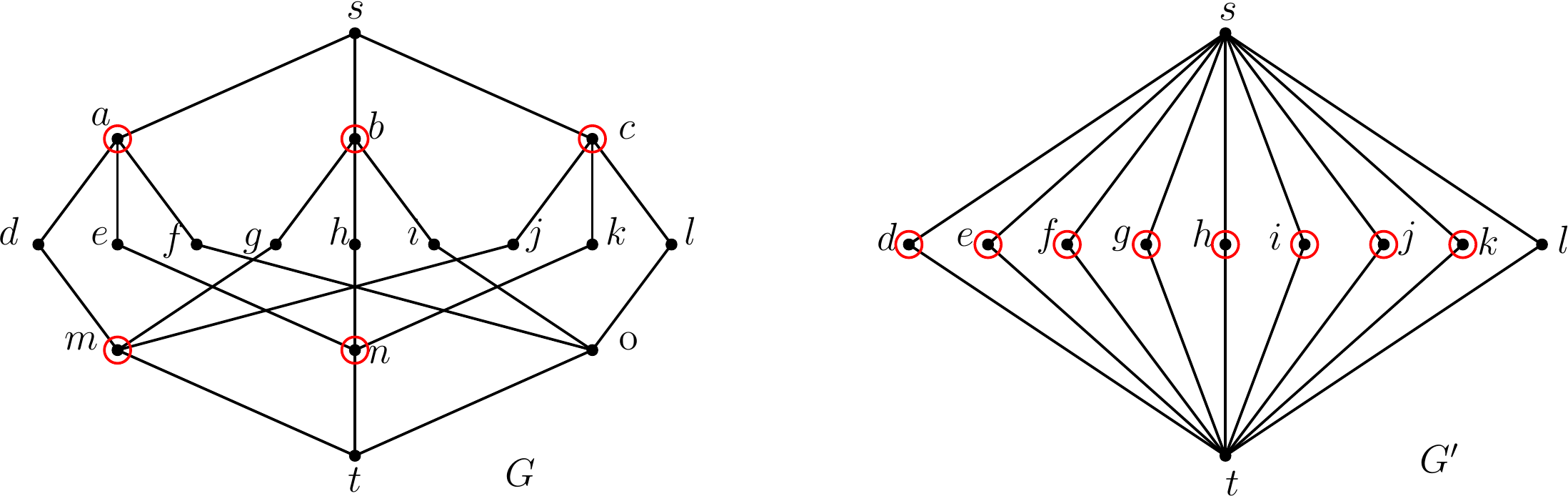}
    \caption{Graph $G'$ is a minor of graph $G$. The set of circled vertices represents a minimum tracking set for both the graphs.}
    \label{fig:minor}
\end{figure}

Observe that graph $G'$ can be obtained from $G$ by contracting the following edges: $sa$, $sb$, $sc$, $mt$, $nt$, $ot$. However note that $G'$ requires at least $8$ trackers while $G$ can be tracked with just $5$ trackers. Since \tp is not closed under minors, the well known \textit{graph minor theorem} does not apply to it~\cite{seymour1,seymour2,seymour-chapter}. Hence, the problem is inherently different from the standard covering problems and the well known techniques of finding some specific obstructions and devising algorithms to hit/cover them, shall not work for the case of \tp.

Even proving that \tp is in \textsc{NP} is non-trivial. See~\cite{tr-j} for details. A combinatorial generalization of \tp is studied in~\cite{caldam18},~\cite{tr1-j}, where the input is a set system, and it is required to find a subset of elements from the universe that have a unique intersection with each set in the family. The problem has been shown to be a dual of the \textsc{test cover} problem. \tp was proven to be polynomial time solvable for chordal graphs, tournament graphs and for the case when edges are used as trackers instead of vertices~\cite{iwoca},~\cite{polytime-arxiv}.
A related problem, \textit{Identifying Path Cover} has been discussed in~\cite{pathcover} and~\cite{identifying-path}. \textit{Identifying Path Cover} requires finding a set of paths that cover all the vertices in a graph and  uniquely identifies each vertex by inclusion in a distinct set of paths. 

\section{Preliminaries}
\label{sec:prelim}

A kernelization algorithm is typically obtained using what are called {\it reduction rules}. These rules transform a given parameterized instance in polynomial time to an equivalent instance, and a rule is said to be {\em safe} if the resulting graph has a tracking set of size at most $k$ if and only if the original instance has one.

Throughout the paper, we assume graphs to have no self loops or multi-edges.
When considering tracking set for a graph $G=(V,E)$, we assume that the given graph is an \stgraph, i.e. the graph contains a unique source $s\in V$ and a unique destination $t\in V$ (both $s$ and $t$ are known), and we aim to find a tracking set that can distinguish between all simple \stpaths. If $a,b\in V$, then unless otherwise stated, $\{a,b\}$ represents the set of vertices $a$ and $b$, and $(a,b)$ represents an edge between $a$ and $b$. For a vertex $v\in V$, \textit{neighborhood} of $v$ is denoted by $N(v)=\{x \mid (x,v)\in E\}$. We use $deg(v)=|N(v)|$ to denote degree of vertex $v$. For a graph $G$, we use $G'\subseteq G$ to denote that $G'$ is a subgraph of $G$.
For a vertex $v\in V$ and a subgraph $G'$, $N_{G'}(v)=N(v)\cap V(G')$. For a subset of vertices $V'\subseteq V$ we use $N(V')$ to denote $\bigcup_{v\in V'} N(v)$. With slight abuse of notation we use $N(G')$ to denote $N(V(G'))$. For a graph $G$ and a set of vertices $S\subseteq V(G)$, $G-S$ denotes the subgraph induced by the vertex set $V(G)\setminus V(S)$. If $S$ is a singleton, we may use $G-x$ to denote $G-S$, where $S=\{x\}$. We use $[m]$ to denote the set of integers $\{1,\dots,m\}$.

For a path $P$, $V(P)$ denotes the vertex set of path $P$, and for a subgraph(or graph) $G'$, $V(G')$ denotes the vertex set of $G'$. For a subgraph(or graph) $G'$, $E(G')$ denotes the edge set of $G'$. Let $P_1$ be a path between vertices $a$ and $b$, and $P_2$ be a path between vertices $b$ and $c$, such that $V(P_1)\cap V(P_2)=\{b\}$. By $P_1 \cdot P_2$, we denote the path between $a$ and $c$, formed by concatenating paths $P_1$ and $P_2$ at $b$. Two paths $P_1$ and $P_2$ are said to be \textit{vertex disjoint} if their vertex sets do not intersect except possibly at the end points, i.e. $V(P_1)\cap V(P_2) \subseteq \{a,b\}$, where $a$ and $b$ are the starting and end points of the paths. By distance we mean length of the shortest path, i.e. the number of edges in that path. 
For a graph $G=(V,E)$, an FVS is a set of vertices $S\subseteq V$ such that $G\setminus S$ is a forest.

\section{Analyzing structures}
\label{sec:structures}

In this section we analyze the problem of distinguishing paths in some specific graph structures along with providing some basic preprocessing steps. We start with some reduction rules and basic results from~\cite{tr-j} followed by additional ones. Later we give some important lemmas based on tree like structures, which forms the base for vertex counting arguments for kernels in subsequent sections.
Following reduction rules are applied exhaustively as long as they are applicable.

\begin{Reduction Rule}~\cite{tr-j}
\label{red:stpath}
If there exists a vertex or an edge that does not participate in any \stpath then delete it.
\end{Reduction Rule}

In the rest of the paper we assume that each vertex and edge participates in at least one \stpath. 

\begin{Reduction Rule}
\label{red:no-deg-one}
If $V\setminus\{s,t\}=\emptyset$, then return a trivial YES instance. Else, if degree of $s$ (or $t$) is $1$ and $N(s)\neq t$ ($N(t)\neq s$), then delete $s$($t$), and label the vertex adjacent to it  as $s$($t$). 
\end{Reduction Rule}

\begin{lemma}
\label{lemma:red-no-deg-one}
Reduction Rule~\ref{red:no-deg-one} is safe and can be implemented in polynomial time.
\end{lemma}
\begin{proof}
Observe that if $V\setminus\{s,t\}=\emptyset$, then $G$ consists of only the edge $(s,t)$. Since there exists only one \stpath, no trackers are needed for distinguishing. Hence, the given instance is a YES instance. Such a case can be identified in constant time.

Else $V\setminus\{s,t\}\neq\emptyset$. Consider a graph $G$ where $N(s)=\{a\}$ ($N(t)=\{a\}$). Note that since all \stpaths pass through $a$, a minimal tracking set need not contain $a$. Also the sequence of each \stpath starts (ends) with $s$ ($t$) followed (preceded) by $a$. Hence we can simply relabel $a$ as the source (destination) vertex $s$ ($t$), delete $s$ ($t$) from the graph $G$, and this would not affect a tracking set for the graph. Note that this reduction rule cannot be applied if $s$ and $t$ are neighbors, since we cant delete either of them from $G$.
\end{proof}

Now we recall two reduction rules from~\cite{ep-planar} that help us bound the number of degree two vertices in the graph.

\begin{Reduction Rule}~\cite{ep-planar}
\label{red:twoinseries}
If there exist $a,b,c\in V(G)$ such that $deg(b)=deg(c)=2$, $b,c\notin\{s,t\}$, and $N(b)=\{a,c\}$, then delete $b$ and introduce an edge between $a$ and $c$.
\end{Reduction Rule}

%

%
\begin{Reduction Rule}~\cite{ep-planar}
\label{red:deg-2-triangle}
If there exist $a,b,c\in V(G)$ such that $N(b)=\{a,c\}$ and $(a,c)\in E(G)$ and $b\notin\{s,t\}$, then mark $b$ as a tracker, delete $b$ from $G$ and set $k=k-1$.
\end{Reduction Rule}

Next we recall a monotonicity lemma and a corollary from~\cite{tr-j}, which says that if a subgraph of $G$ cannot be tracked with $k$ trackers, then $G$ cannot be tracked with $k$ trackers either.

\begin{lemma}~\cite{tr-j}
Let $G=(V,E)$ be a graph and $G'=(V',E')$ be a subgraph of $G$ such that $\{s,t\} \in V'$. If $T$ is a tracking set for $G$ and $T'$ is a minimum tracking set for $G'$, then $|T'| \leq |T|$.
\label{subgraph}
\end{lemma}

\begin{definition}
Let $G'\subseteq G$ be a subgraph. If $u,v\in V(G')$ is a pair of distinct vertices then, for the subgraph $G'$, $u$ is a \textbf{local source} and $v$ is a \textbf{local destination} if the following hold: \textsf{(a)} there exists a path in $G$ from $s$ to $u$, say $P_{su}$, and another path from $v$ to $t$, say $P_{vt}$, \textsf{(b)} $V(P_{su})\cap V(P_{vt})=\emptyset$, \textsf{(c)} $V(P_{su})\cap V(G')=\{u\}$ and $V(P_{vt})\cap V(G')=\{v\}$.
\end{definition}
Note that a subgraph may have multiple local source-destination pairs. 

\begin{lemma}~\cite{tr-j}[Rephrased]
If each vertex and edge in graph $G$ participate in an \stpath, then for a subgraph $G'\subseteq G$ containing at least one edge, it holds that $G'$ contains a local source and a local destination.
\label{lemma:induced}
\end{lemma} 

In the rest of the paper the phrase `subgraph cannot be tracked by $k$ trackers' implies that the paths between a local source and destination in a subgraph cannot be tracked with $k$ trackers.
Due to Lemma~\ref{subgraph} and Lemma~\ref{lemma:induced}, we have the following two corollaries.

\begin{corollary}~\cite{tr-j}
\label{corollary:subgraph}
If a subgraph of $G$ that contains both $s$ and $t$ cannot be tracked by $k$ trackers, then $G$ cannot be tracked by $k$ trackers either.
\end{corollary}

\begin{corollary}
\label{corollary:subgraph-pair}
If there exists a subgraph $G'$ of $G$, and there exists a pair of vertices $u,v\in V(G')$, such that $u$ is a local source for $G'$ and $v$ is a local destination for $G'$, and all paths between $u$ and $v$ in $G'$ cannot be tracked by at most $k$ trackers, then $G$ cannot be tracked by at most $k$ trackers.
\end{corollary}

 Next corollary forms a starting point for the kernelization algorithms. 

\begin{corollary}~\cite{tr-j}
The size of a minimum tracking set $T$ for $G$ is at least the size of a minimum FVS for $G$.
\label{corollary:kfvs}
\end{corollary}

For the rest of the paper, we assume that the input graph has already been preprocessed using Reduction Rules~\ref{red:stpath} to \ref{red:deg-2-triangle} and hence the following hold:

\begin{enumerate}

\item All vertices and edges in the graph participate in some \stpath.

\item Degree of all vertices in the graph is at least two, and each vertex of degree two has both its neighbors with degree three or higher.

\item There exist at least two \stpaths in the graph.

\end{enumerate}

\subsection{Vertex Disjoint Paths}

Here we give a bound on the number of vertex disjoint paths that can exist between a pair of vertices in a graph $G$, given that $G$ can be tracked with at most $k$ trackers. While in~\cite{tr-j} it is proven that there can exist at most $k+3$ vertex disjoint paths between a pair of vertices in $G$, we improve the bound to $k+1$. The new bound allows easy analysis and computation in future lemmas.

\begin{lemma}
\label{lemma:disjoint-simple}
If there exists two vertices $u,v\in V$ such that there exists more than $k+1$ vertex disjoint paths between $u$ and $v$, and the graph induced by these $k+1$ paths along with $u$ and $v$ has $u$ as a local source and $v$ as a local destination, then $G$ cannot be tracked with at most $k$ trackers.
\end{lemma}

\begin{proof}
For a contradiction assume that there exist $k+2$ vertex disjoint paths $\mathcal{P}=\{P_1,\hdots,P_{k+2}\}$ between a pair of vertices $u$ and $v$ in $G$, and $G$ can be tracked with at most $k$ trackers. Let $G'$ be the subgraph induced by $V(\mathcal{P})\cup\{u,v\}$. Since $u$ is a local source for $G'$ and $v$ is a local destination for $G'$, there exists a path $P_{su}$ that starts at $s$ and ends at $u$ and does not contain any vertex from $G'-\{u,v\}$, and there exists a path $P_{vt}$ that starts at $v$ and ends at $t$ and does not contain any vertex from $G'-\{u,v\}$. See Figure~\ref{fdisjoint:1}.
Consider a pair of paths $P_i,P_j\in\mathcal{P}$. Let $P_i$ and $P_j$ do not contain any trackers (except for $u$ and $v$). Now consider the \stpaths $P'_i=P_{su}\cdot P_i \cdot P_{vt}$ and $P'_j=P_{su}\cdot P_j \cdot P_{vt}$. Note that $P'_i$ and $P'_j$ contain the same sequence of trackers. Since these paths differ only in the vertices on paths $P_i$ and $P_j$, at least one vertex (except $u$ and $v$) either on $P_i$ or $P_j$ has to be a tracker. Thus as long as there are two paths in $\mathcal{P}$ without any trackers, there will be two \stpaths with same sequence of trackers. Hence, at least $k+1$ paths in $\mathcal{P}$ need a tracker on them. Due to Corollary~\ref{corollary:subgraph-pair}, we know that if $G'$ cannot be tracked with at most $k$ trackers then $G$ cannot be tracked with at most $k$ trackers. This contradicts the initial assumption, and hence completes the proof.
\end{proof} 

\begin{figure}
\begin{minipage}[h]{0.45\textwidth}
\centering
  {\includegraphics[scale=0.7]{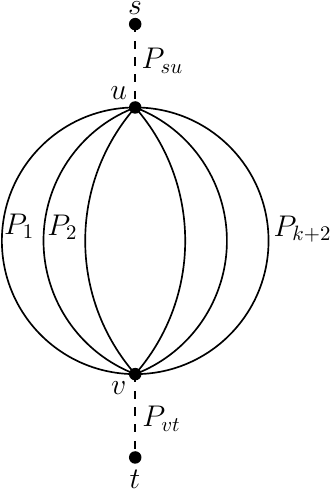}}  
     \caption{Vertex disjoint paths between a local source and destination}
     \label{fdisjoint:1}
\end{minipage}\hspace{10mm}
\begin{minipage}[h]{0.45\textwidth}
\centering
  {\includegraphics[scale=0.7]{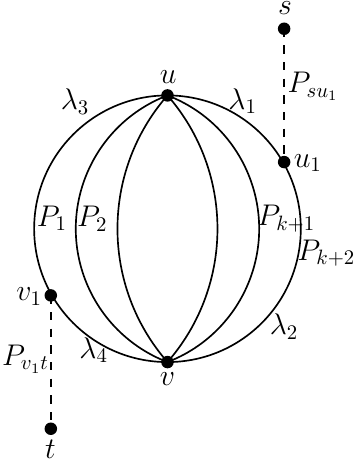}} 
  \vspace{4mm}
    \caption{Vertex disjoint paths between a pair of vertices}
    \label{fdisjoint:2}
\end{minipage}
\end{figure}

\begin{lemma}
If there exists two vertices $u,v \in V$ such that there exists more than $k+1$ vertex disjoint paths between $u$ and $v$, then $G$ cannot be tracked with at most $k$ trackers.
\label{lemma:disjoint}
\end{lemma}

\begin{proof}
For a contradiction assume that there exist $k+2$ vertex disjoint paths $\mathcal{P}=\{P_1,\hdots,P_{k+2}\}$ between a pair of vertices $u$ and $v$ in $G$, and $G$ can be tracked with at most $k$ trackers. Let $G'$ be the subgraph induced by $V(\mathcal{P})\cup\{u,v\}$. Due to Lemma~\ref{lemma:induced}, there exists a local source, say $u_1$, and a local destination, say $v_1$, in $G'$. We consider various cases possible based on position of $u_1$ and $v_1$ in $G'$.
\begin{itemize}

\item When $u=u_1$ and $v=v_1$, or $u=v_1$ and $v=u_1$. Both of these cases are symmetric to each other, and have been proven in Lemma~\ref{lemma:disjoint-simple}.

\item When $\{u,v\}\cap\{u_1,v_1\}=\emptyset$. First we consider the case when $u_1$ and $v_1$ lie on different paths in $\mathcal{P}$. See Figure~\ref{fdisjoint:2}. We denote the path between $u$ and $u_1$ (subpath of $P_{k+2}$) by $\lambda_1$, the path between $u_1$ and $v$ (subpath of $P_{k+2}$) by $\lambda_2$, the path between $u$ and $v_1$ (subpath of $P_1$) by $\lambda_3$, and the path between $v$ and $v_1$ (subpath of $P_1$) by $\lambda_4$. We denote the paths in $\mathcal{P}\setminus \{P_1,P_{k+2}\}$ by $\mathcal{P}'$. Any \stpath in $G$ that passes through $G'$ will be one among the following types:
\begin{enumerate}

\item $P_{su_1} \cdot \lambda_1 \cdot P_i \cdot \lambda_4 \cdot P_{v_1t}$, where $P_i\in \mathcal{P}'$
\item $P_{su_1} \cdot \lambda_2 \cdot P_i \cdot \lambda_3 \cdot P_{v_1t}$, where $P_i\in \mathcal{P}'$
\item $P_{su_1} \cdot \lambda_1 \cdot \lambda_3 \cdot P_{v_1t}$
\item $P_{su_1} \cdot \lambda_2 \cdot \lambda_4 \cdot P_{v_1t}$

\end{enumerate}
Consider the first two cases. Let $G''$ be the graph induced by $\mathcal{P}'$. Observe that $u$ and $v$ are local source and destination for $G''$, since there exists a path $P_{su_1}\cdot \lambda_1$ from $s$ to $u$, and a path $P_{v_1t}\cdot\lambda_4$ from $v$ to $t$, and these paths intersect with $G''$ only at $u$ and $v$. Since there are $k$ paths between $u$ and $v$ in $G''$, due to Lemma~\ref{lemma:disjoint-simple}, these require at least $k-1$ trackers in $V(\mathcal{P}')\setminus\{u,v\}$. If each of the $k$ paths in $\mathcal{P}'$ has a tracker, the paths indicated in cases 3,4 have the same sequence of trackers, and this contradicts the assumption that $G$ has a tracking set of size $k$.
Else, without loss of generality, let $P_{k+1}$ be the path in $\mathcal{P}'$ that is left without a tracker.

Cases 3,4 denote two vertex disjoint paths between $u_1$ and $v_1$ along $P_1$ and $P_{k+2}$. Hence, due to Lemma~\ref{lemma:disjoint-simple}, there must be a tracker on either $V(\lambda_1)\cup V(\lambda_3)\setminus \{u_1,v_1\}$ or $V(\lambda_2)\cup V(\lambda_4)\setminus \{u_1,v_1\}$. We consider following cases:
\begin{itemize}

\item There exists a tracker in $V(\lambda_1)\setminus \{u_1,v_1\}$: Paths $P_{su_1} \cdot \lambda_1 \cdot \lambda_3 \cdot P_{v_1t}$ and $P_{su_1} \cdot \lambda_1 \cdot P_{k+1} \cdot \lambda_4 \cdot P_{v_1t}$ contain the same set of trackers.

\item There exists a tracker in $V(\lambda_2)\setminus \{u_1,v_1\}$: Paths $P_{su_1} \cdot \lambda_2 \cdot P_{k+1} \cdot \lambda_3 \cdot P_{v_1t}$ and $P_{su_1} \cdot \lambda_2 \cdot \lambda_4 \cdot P_{v_1t}$ contain the same set of trackers.

\item There exists a tracker in $V(\lambda_3)\setminus \{u_1,v_1\}$: Paths $P_{su_1} \cdot \lambda_1 \cdot \lambda_3 \cdot P_{v_1t}$ and $P_{su_1} \cdot \lambda_2 \cdot P_{k+1} \cdot \lambda_3 \cdot P_{v_1t}$ contain the same set of trackers.

\item There exists a tracker in $V(\lambda_4)\setminus \{u_1,v_1\}$: Paths $P_{su_1} \cdot \lambda_1 \cdot P_{k+1} \cdot \lambda_4 \cdot P_{v_1t}$ and $P_{su_1} \cdot \lambda_2 \cdot \lambda_4 \cdot P_{v_1t}$ contain the same set of trackers.

\end{itemize}
 
All the above cases contradict the assumption that $G$ can be tracked with at most $k$ trackers.

Next we consider the case when both $u_1$ and $v_1$ lie on the same path in $\mathcal{P}$. Without loss of generality assume that $u_1$ and $v_1$ both lie on $P_1$. Here note that there exists one path between $u_1$ and $v_1$ that is a strict subpath of $P_1$, and the remaining paths between $u_1$ and $v_1$ pass through $\mathcal{P}\setminus P_1$, via vertices $u$ and $v$. Observe that $u$ and $v$ are a local source and destination for the subgraph $G'''$ induced by $V(\mathcal{P}\setminus P_1)$. Since there are $k+1$ paths between $u$ and $v$ in the subgraph $G'''$, due to Lemma~\ref{lemma:disjoint-simple}, at least $k$ trackers are required in $V(G''')$. If there are $k+1$ trackers in $V(G''')$, it contradicts the assumption that $G$ can be tracked with at most $k$ trackers. If there are $k$ trackers in $V(G''')$, without loss of generality, let $P_2$ be the path without any tracker. Now observe that there are two paths between $u_1$ and $v_1$, the one that does not pass through $u$ and $v$ (subpath of $P_1$) and the one that passes through $u$ and $v$, through $P_2$, that do not have any trackers on them. Due to Lemma~\ref{lemma:disjoint-simple}, at least one tracker is required on one of these paths. Hence we have a contradiction to the assumption that the graph can be tracked with at most $k$ trackers.

\item When $u=u_1$ and $v\neq v_1$, or, $u\neq u_1$ and $v=v_1$. 
Consider $u=u_1$, and $v\neq v_1$. This case can be argued similar to the second case, except that now $\lambda_1=u=u_1$. Similarly, if $u\neq u_1$, and $v=v_1$, the case is similar to the second case, except that now $\lambda_4=v=v_1$.

\end{itemize}

Note that the case when $u_1=v_1$, is not possible after  application of Reduction Rule~\ref{red:stpath}.
Correctness of the proof follows from Corollary~\ref{corollary:subgraph}. Hence the lemma holds.
\end{proof}

Next we give a reduction rule based on Lemma~\ref{lemma:disjoint}.

\begin{Reduction Rule}
\label{red:disjoint-m-paths}
Let $G'$ be a subgraph of $G$, consisting of a pair of vertices $a,b$ adjacent to $i$ vertices, each of degree two. If  $a$ is a local source for $G'$ and $b$ is a local destination for $G'$, then arbitrarily mark $i-1$ of the $i$ vertices of degree two as trackers and delete them. If $i> k+1$ return a NO instance, else set $k=k-i-1$.
\end{Reduction Rule}

\begin{lemma}
\label{lemma:red-disjoint-m-paths}
Reduction Rule~\ref{red:disjoint-m-paths} is safe and can be implemented in polynomial time.
\end{lemma}
\begin{proof}
Let $V_i$ be the set of $i$ vertices of degree two that are adjacent to $a$ and $b$ and let $V_{i-1}$ be the set of $i-1$ vertices that were marked as trackers and deleted. Let $G'$ be the newly created graph after the deletion of $V_{i-1}$. We claim that $G'$ has a tracking set of size $k-i-1$ if and only if $G$ has a tracking set of size $k$. Suppose $G'$ has a tracking set $T'$ of size $k-i-1$. If we add the vertices of $V_{i-1}$ back to $G'$ along with their edges, there exists $i$ vertex disjoint paths between $a$ and $b$. Since $a$ and $b$ are the local source and destination, due to Lemma~\ref{lemma:disjoint} at least $i-1$ trackers are required on the vertices of $V_i$. We mark all the vertices in $V_{i-1}$ as trackers. Now all paths between $a$ and $b$ are tracked. Since all other paths were already being tracked by $T'$ in $G'$, $T'\cup V_{i-1}$ is a tracking set of size $k$ for $G$.

In the other direction let $T$ be a tracking set of size $k$ in $G$. Let $x\in V_i\setminus V_{i-1}$. We claim that $G'$ has a tracking set of size $k-i-1$. Suppose not. Then there exists two \stpaths, say $P_1$ and $P_2$, in $G'$ that have the same sequence of trackers. Observe that if both $P_1$ and $P_2$ do not intersect with edges $(a,x)$ and $(x,b)$, then $T$ is not a tracking set of size $k$ in $G$. This implies that $P_1$ and $P_2$ are also two paths with same sequence of trackers in $G$. Note that the trackers on vertices in $V_{i-1}$ cannot be used to distinguish between $P_1$ and $P_2$, as that would leave some untracked paths between $a$ and $b$. Thus $T$ is not a tracking set for $G$, which is a contradiction. This completes the proof of safeness of Reduction Rule~\ref{red:disjoint-m-paths}.

In order to implement Reduction Rule ~\ref{red:disjoint-m-paths}, we consider each pair of vertices $u,v\in V(G)$, and compare all their neighbors, to check for common neighbors of degree two. This can be done in $\Oh(n^4)$ time. Hence the rule can be applied in polynomial time.
\end{proof}

\subsection{Tree-sink structure}
\label{subsec:tree-sink}

In this section we discuss a specific graph structure, namely the \textit{tree-sink structure}, and prove a lower bound for the number of trackers required if such a structure exists in an \stgraph.

\begin{definition}
A \textbf{tree-sink structure} in a graph $G$ is a subgraph $G'$ such that $V(G')=V(Tr)\cup\{x\}$, where $Tr$ is a tree with at least two vertices, and all of its leaves are adjacent to $x$. Here $Tr$ is the \textbf{tree} while $x$ is the \textbf{sink} of the tree-sink structure.
\end{definition}
Note that $x$ may or may not be adjacent to the non-leaf vertices of $Tr$. See Figure~\ref{fig:tree-sink}.
We prove that if the sink $x$ is adjacent to more than $k+1$ vertices in $Tr$, then $G$ cannot be tracked with at most $k$ trackers. 
We start with a simple case when the graph $G$ itself is a tree-sink structure and either $s$ or $t$ is the sink.


\begin{figure}[ht]
\centering
\includegraphics[scale=0.3]{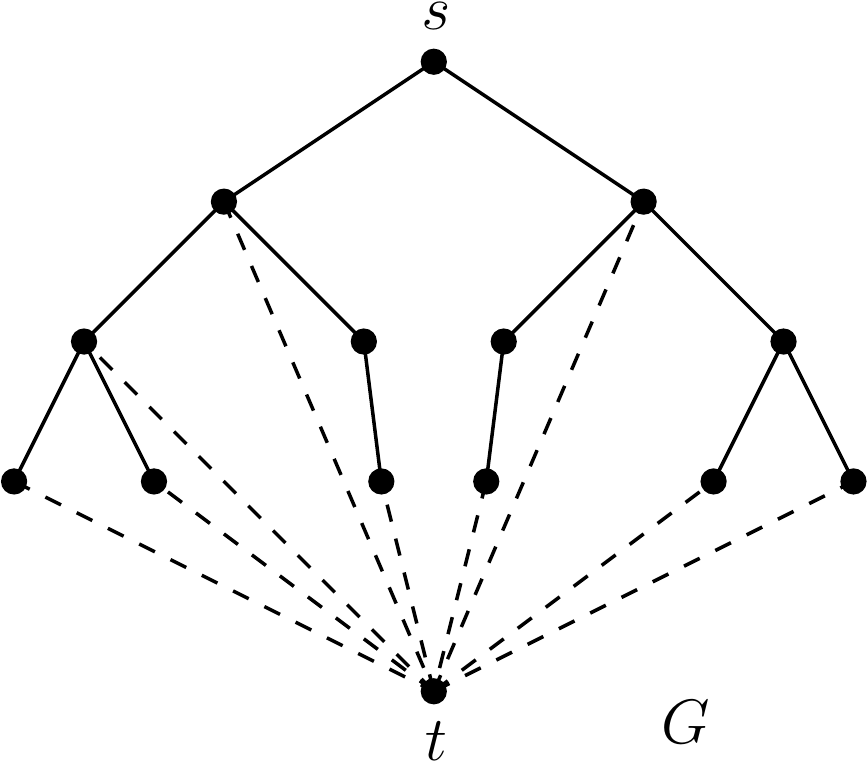} 
\caption{Tree-sink structure: Solid lines are edges of tree, and dashed lines are edges between the sink $t$ and vertices of the tree.} 
\label{fig:tree-sink}
\end{figure}

\begin{lemma}
\label{lemma:rooted-tree}
Let $G$ be an \stgraph that forms a tree-sink structure with $x\in V(G)$ as the sink and $x\in\{s,t\}$. If $|N(x)|=\delta$, then $\delta-1$ trackers are required in $G$, and these trackers have to be in $V(Tr)$, where $Tr$ is the tree induced by $G-x$.
\end{lemma}

\begin{proof}
Without loss of generalization we assume that $x=t$. We root $Tr$ at the source vertex $s$. 
Consider that graph $G$ has already been preprocessed using Reduction Rules~\ref{red:stpath}, \ref{red:no-deg-one} and \ref{red:twoinseries}.

We prove the lemma by induction on the value of $\delta$.
Observe that due to Reduction Rule~\ref{red:stpath}, \ref{red:no-deg-one} and \ref{red:twoinseries}, $\delta=1$ is not possible.
Thus the base case for induction is when $\delta=2$. Note that in this case $G$ is either a triangle or a four cycle. See Figure~\ref{fig:two-neighbors}.
\begin{figure}[ht]
\centering
\includegraphics[scale=0.35]{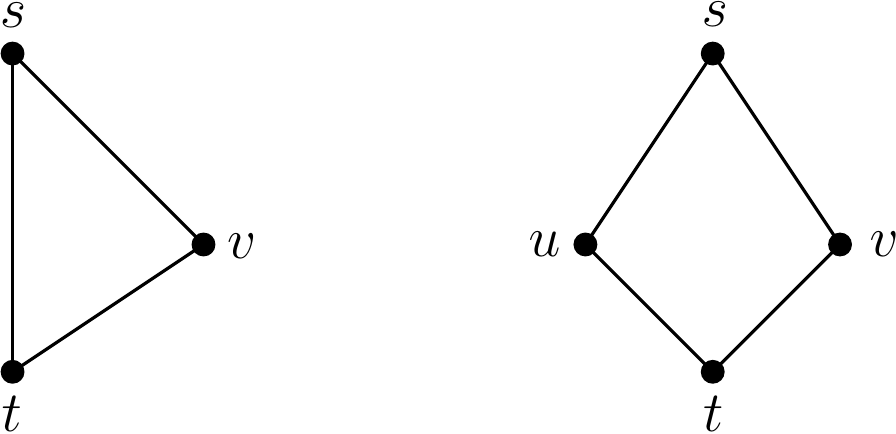} 
\caption{Possible cases when the sink $t$ has two neighbors in the tree induced by $G-t$.} 
\label{fig:two-neighbors}
\end{figure}
Consider the case when $G$ is a triangle. Due to Reduction Rule~\ref{red:deg-2-triangle}, the vertex $v\in V(G)\setminus\{s,t\}$ is marked as a tracker and deleted. Consider the case when $G$ is a four cycle. Observe that there exist two vertices, say $u,v$, of degree two each, adjacent to $s$ and $t$. Due to Reduction Rule~\ref{red:disjoint-m-paths}, one among $u$ and $v$ is marked as a tracker and deleted. Note that in both the cases, after application of the corresponding reduction rules, $G$ comprises only of the edge $(s,t)$. Due to Reduction Rule~\ref{red:no-deg-one}, this is a trivial YES instance. Hence, when $\delta=2$, exactly one tracker is required in $G$. This proves that the claim holds for the base case.

Next, for induction hypothesis, we assume that the claim holds for $\delta=i$, i.e. if the sink is adjacent to $i$ vertices, then $i-1$ trackers are required in $G$. Consider the case when $\delta=i+1$. Note that here $\delta\geq 3$. Due to Reduction Rule~\ref{red:stpath}, all leaves in $Tr$ are adjacent to $t$, $Tr$ being the tree induced by $G-t$. Consider a leaf vertex, say $v_1\in V(Tr)$, that is at maximum distance from $s$. Since $deg(v_1)=2$, due to Reduction Rule~\ref{red:twoinseries}, the degree of its parent node, say $v_p$, is at least $3$. Thus either $v_p$ has another child node, or $v_p$ is adjacent to $t$. We analyze both the possibilities:
\begin{itemize}
\item \textit{Case I}: $v_p$ has another child node, say $v_2$. Since $v_1$ is at maximum distance possible from $s$, $v_2$ is a leaf node in $Tr$. Observe that the graph $G'$ induced by $v_1$, $v_2$, $v_p$ and $t$ has $v_p$ as a local source and $t$ as a local destination, and $deg(v_1)=deg(v_2)=2$. Due to Reduction Rule~\ref{red:disjoint-m-paths}, either $v_1$ or $v_2$ will be marked a tracker and deleted. This reduces the value of $\delta$ from $i+1$ to $i$, while using one tracker.

\item \textit{Case II}: $v_p$ is adjacent to $t$. Observe that $v_1$, $v_p$ and $t$ form and triangle and $deg(v_1)=2$. Due to Reduction Rule~\ref{red:deg-2-triangle}, $v_1$ will be marked as a tracker and deleted. This reduces the value of $\delta$ from $i+1$ to $i$, while using one tracker.

\end{itemize}
Now $\delta=i$. Due to induction hypothesis, we know that when $\delta=i$, then $i-1$ trackers are required in $G$. Since we already used a tracker in both the above cases, the total number of trackers required when $\delta=i+1$, is $i$. Since the sink is $t$ itself, all the trackers need to be in $V(Tr)$. This completes the proof.
\end{proof}

Next we give a corollary which makes the above lemma more usable for the sake of our future arguments.
\begin{corollary}
\label{corollary:rooted-tree-general}
Let $G$ be a graph and $G'$ be a subgraph of $G$ such that $G'$ induces a tree-sink structure with $v\in V(G')$ as its sink. If $|N(v)\cap V(G')|=\delta$, and $v$ is either a local source or a local destination for $G'$, then the size of a tracking set for $G$ is at least $\delta - 1$. Further these $\delta-1$ trackers need to be in $V(G')\setminus\{v\}$.
\end{corollary}

\begin{proof}
Consider the subgraph $G'$. Without loss of generality, we assume that $v$ is a local destination for $G'$. Let $u\in V(G')$ be a local source corresponding to the local destination $v$.
Due to Lemma~\ref{lemma:rooted-tree}, we have that $\delta-1$ trackers are required in $V(G')\setminus\{v\}$ to track all paths between $u$ and $v$. From Corollary~\ref{corollary:subgraph-pair}, if in a subgraph all paths between a local source and destination cannot be tracked with $k$ trackers then the graph cannot be tracked with $k$ trackers. Hence if $k<\delta-1$, then $G$ cannot be tracked with at most $k$ trackers. Thus the size of a tracking set for $G$ is at least $\delta-1$. It follows from Lemma~\ref{lemma:rooted-tree} that these $\delta - 1$ trackers need to be in $V(G')\setminus\{v\}$.
\end{proof}

The next lemma generalizes the result in Corollary~\ref{corollary:rooted-tree-general}. We prove that regardless of where $s$ and $t$ lie in graph $G$, if $G$ forms a tree-sink structure, then the size of the tracking set for $G$ is at least the number of neighbors of the sink in the tree minus one.

\begin{lemma}
\label{lemma:unrooted-tree}
If an \stgraph $G$ forms a tree-sink structure such that $x\in V(G)$ is the sink and $G-x$ induces a tree and $|N(x)|=\delta$, then the size of a tracking set for $G$ is at least $\delta-1$, and at least $\delta-2$ trackers are required in $G-x$.
\end{lemma}

\begin{proof}
Let $Tr$ be the tree induced by $V(G\setminus\{x\})$.
The case when $x\in\{s,t\}$ has been proven in Lemma~\ref{lemma:rooted-tree}.
Consider the case when $s,t \in V(Tr)$.
We start by rooting the tree at $s$. Now create a graph $G'$ by removing the edge between $t$ and its parent vertex, say $\hat{t}$, in $Tr$. Observe that in $G'$, there exists a tree, say $Tr_1$ that can be considered rooted at $s$, consisting of all those vertices in $V(Tr)$ that are not descendants of $t$ in $Tr$, with all its leaves adjacent to the vertex $x$. There exists another tree, say $Tr_2$, rooted at $t$, consisting of all of its descendants in $Tr$, with all of its leaves adjacent to $x$. See Figure~\ref{fig:tree-gen}. We denote the graph induced by $V(Tr_1)\cup \{x\}$ by $G_1$, and the graph induced by $V(Tr_2)\cup \{x\}$ by $G_2$. Let $\delta_1$ be the number of leaves in $Tr_1$, and $\delta_2$ be the number of leaves in $Tr_2$. Note that $\delta_1+\delta_2=\delta$.

\begin{figure}[ht]
\centering
\includegraphics[scale=0.34]{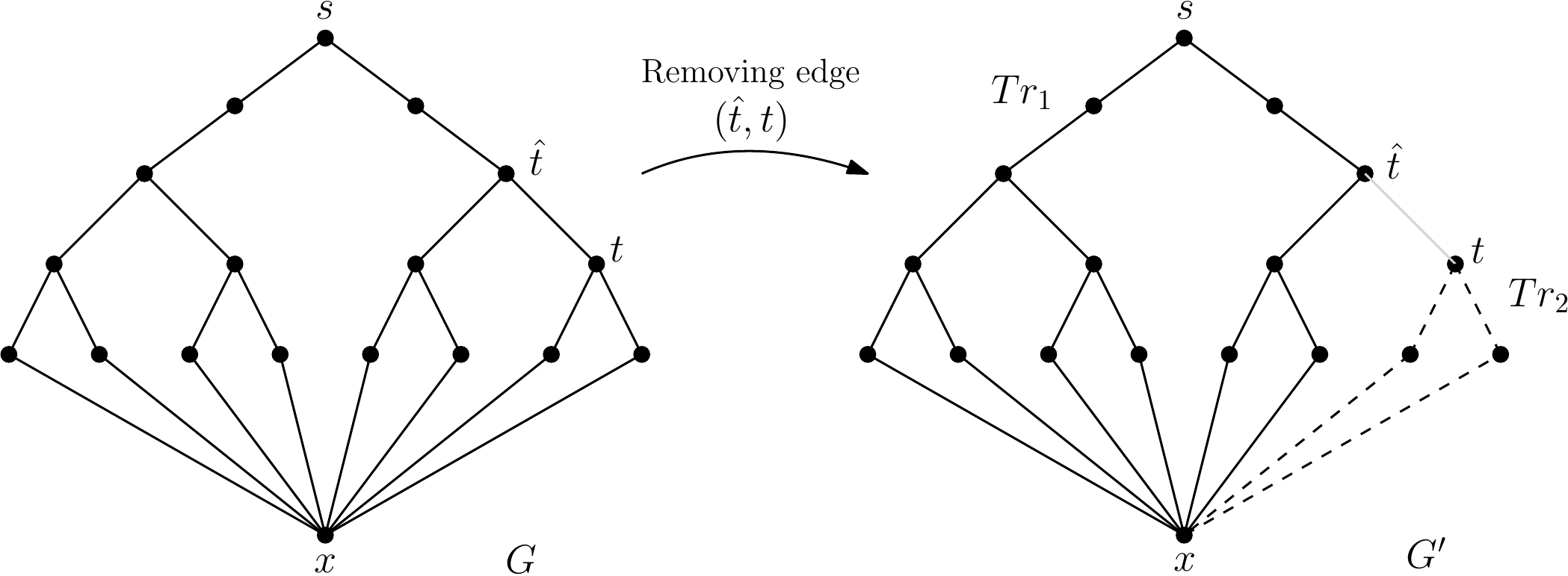} 
\caption{Removing edge $(\hat{t},t)$ from $G$ creates two tree-sink structures in $G'$, with trees $Tr_1$ (shown with solid lines) and $Tr_2$ (shown in dashed lines) and sink $x$.} 
\label{fig:tree-gen}
\end{figure}

Note that $x$ is a local destination for $G_1$. Hence by Corollary~\ref{corollary:rooted-tree-general}, since $Tr_1$ has $\delta_1$ many leaves, the size of a tracking set for $G$ is at least $\delta_1-1$, and all these trackers must be in $V(Tr_1-x)$.

Note that $x$ is a local source for $G_2$. Hence by Corollary~\ref{corollary:rooted-tree-general}, since $Tr_2$ has $\delta_2$ many leaves, the size of a tracking set for $G$ is at least $\delta_2-1$, and all these trackers must be in $V(Tr_2-x)$.

If there exists at least $\delta_1+\delta_2-1$ trackers in $G$, then the lemma holds. Else there exist $\delta_1-1$  trackers in $V(Tr_1-x)$ and $\delta_2-1$ trackers in $V(Tr_2-x)$. Hence, there exists exactly one path in $G_1$, say $P_1$, from $s$ to $x$ that does not contain any trackers, and exactly one path in $G_2$, say $P_2$, from $x$ to $t$ that does not contain any trackers. Consider the path $P'=\{s\}\cdot P_1\cdot \{x\}\cdot P_2\cdot\{t\}$. Note that if $G$ contains a total of $\delta_1+\delta_2-2$ trackers, then $x$ is not a tracker and hence $P'$ does not contain any trackers. Recall the edge $e$ that was initially removed between $t$ and its parent, $\hat{t}$, in $Tr$. Consider the path in $G_1$ from $s$ to $\hat{t}$, say $P_{s\hat{t}}$. We consider the following two scenarios.

\begin{itemize}
\item $P_{s\hat{t}}$ is a subpath of the path $P_1$. Consider the paths $\{s\} \cdot P_1\cdot\{x\}\cdot P_2\cdot\{t\}$, and $\{s\}\cdot P_{s\hat{t}}\cdot\{t\}$. Observe that both these paths have no trackers. Hence one more tracker is needed, either in $V(P_1)$ or $V(P_2)$ in order to distinguish them in $G$. 
\item $P_{s\hat{t}}$ is not a subpath of the path $P_1$. If $P_{s\hat{t}}$ does not have a tracker, both the paths $\{s\} \cdot P_1\cdot\{x\}\cdot P_2\cdot\{t\}$ and $\{s\}\cdot P_{s\hat{t}}\cdot\{t\}$ do not contain any trackers. If $P_{s\hat{t}}$ has a tracker, let $t_r\in V(P_{s\hat{t}})$ be the tracker that is closest to $\hat{t}$. Since $\delta-1$ is the minimum number of trackers required in $Tr_1$, there exists a path from $t_r$ to $x$ (and not passing through $s$) in $G_1$ that does not contain any trackers. Lets denote this path by $P_{t_r x}$. Let $P_{st_r}$ be the path from $s$ to $t_r$ that is a subpath of $P_{s\hat{t}}$. Now observe that paths $\{s\}\cdot P_{st_r}\cdot P_{t_r x} \cdot\{x\} \cdot P_2$ and $\{s\} \cdot P_{s\hat{t}}\cdot \{t\}$ have the same set of trackers. Hence in both the cases discussed one more tracker is required in $G$.
\end{itemize}

Thus the total number of required trackers in $G$ is at least $\delta_1+\delta_2-2+1$, i.e. $\delta-1$. 
Since the sink can be a tracker as well, a tree-sink structure requires at least $\delta-2$ trackers in the vertex set of the tree.
\end{proof}

Lemma~\ref{lemma:unrooted-tree} along with Corollary~\ref{corollary:subgraph} gives us the following corollary.


\begin{corollary}\label{corollary:tree-generalized}
In a graph $G$, if there exists a subgraph $G'$ and a vertex $v\in V(G')$, such that $G'$ forms a tree-sink structure with $v$ as a sink, and $|N(v)\cap V(G')|=\delta$ then the size of a tracking set for $G$ is at least $\delta-1$. Further at least $\delta-2$ trackers are required to be in $V(G')\setminus \{v\}$.
\end{corollary}

For the rest of the paper, we assume that the input graph has already been preprocessed by the reduction rules stated so far.

\section{Quadratic Kernel for General Graphs}
\label{sec:quadratic}

In this section we show that an instance $(G,k)$ of \tp can be reduced to an equivalent instance $(G',k')$ such that if $(G,k)$ is an YES instance then $|V(G')|=\Oh(k'^2)$, $|E(G')|=\Oh(k'^2)$ and $k'\leq k$. We start by applying Reduction Rules~\ref{red:stpath}, \ref{red:no-deg-one}, \ref{red:twoinseries} and \ref{red:deg-2-triangle}. If the instance is not termed a NO instance by any of the reduction rules, we proceed with following.
Recall from Corollary~\ref{corollary:kfvs}, that the size of a minimum tracking set $T$ for $G$ is at least the size of a minimum FVS for $G$. We start by finding a $2$-approximate \fvs $S$, using~\cite{twofvs}.
From Corollary~\ref{corollary:kfvs}, we have the following reduction rule.

\begin{Reduction Rule}~\cite{tr-j}
Apply the algorithm of  {\rm \cite{twofvs}} to find a $2$-approximate solution, $S$ for \fvs. If $|S| > 2k$, then return that the given instance is a NO instance.
\end{Reduction Rule}

\noindent
Observe that $\mathcal{F}=G\setminus S$ is a forest.
Now we try to bound the number of vertices and edges in $\mathcal{F}$ for the case when all \stpaths in $G$ can be tracked with at most $k$ trackers. 
In general, by `tree' we mean a tree in the forest $\mathcal{F}$.
When referring to a tree-sink structure, by `tree' we mean the tree that forms the tree-sink structure.

We give some counting arguments to bound the vertices in $\mathcal{F}$ and the edges incident on these vertices.
We use the notation $Tr\in\mathcal{F}$ to denote that $Tr$ is a tree in the forest $\mathcal{F}$. We assume that $Tr$ is a maximal tree i.e. there does not exist another tree $Tr'$ in $\mathcal{F}$ such that $V(Tr)\subset V(Tr')$. Now we categorize the vertices in $\mathcal{F}$ as follows:

\begin{itemize}

\item $V_1=\{ v\in\mathcal{F} \mid \exists u\in V(Tr), \textrm{ where } Tr \in \mathcal{F}, v\in V(Tr),  N_S(u)\cap N_S(v)\neq\emptyset \}$

\item $V_2=\{ v\in\mathcal{F} \mid \exists u\in V(Tr), \textrm{ where } Tr \in \mathcal{F},  v\notin V(Tr),  N_S(u)\cap N_S(v)\neq\emptyset \}$

\item $V_3=\{ v\in\mathcal{F} \mid \nexists u\in \mathcal{F} \textrm{ and } N_S(u)\cap N_S(v)=\emptyset \}$

\item $V_4=\{v\in\mathcal{F} \mid N(v)\cap S=\emptyset\}$
\end{itemize}

$E_i$ denotes the set of edges between the set of vertices $V_i$ and $S$, where $i\in[3]$.
Note that some vertices in $\mathcal{F}$ may belong to more than one of the above mentioned sets. While giving the counting arguments, we may allow this possible over counting since it does not change the asymptotic value of the bound on $V(\mathcal{F})$. Note that since each vertex in $S$ can have at most one vertex from $V_3$ adjacent to it, and $|S|\leq 2k$, it follows that $|V_3|\leq 2k$.
The total number of vertices in $\mathcal{F}$ will be less than or equal to $|V_1|+|V_2|+|V_3|+|V_4|$.
Now we explore each of the above categories in detail.

\subsection{Bounding $V_1$}
\label{subsec:V1}

Recall that $V_1=\{ v\in\mathcal{F} \mid \exists u\in V(Tr), \textrm{ where } Tr \in \mathcal{F}, v\in V(Tr),  N_S(u)\cap N_S(v)\neq\emptyset \}$.
Hence $V_1$ is the set of vertices that share a neighbor in $S$ with another vertex in the same tree.
We first give a lemma that bounds the number of trees that can form a tree-sink structure with a common sink. This in turn helps us bound the number of trees in $\mathcal{F}$ whose vertices can form tree-sink structures with the vertices in $S$ as sinks.

\begin{lemma}
\label{lemma:flower}
Let there be a vertex $x$ such that $x$ is a sink for $r\geq 2$ disjoint (except for the sink) tree-sink structures, then the numbers of trackers required is at least $r$.
\end{lemma}

\begin{figure}[ht]
\centering
\includegraphics[scale=0.7]{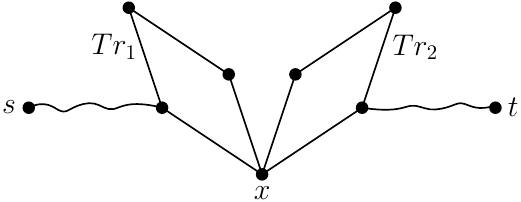} 
\caption{Two tree-sink structures with a common sink.} 
\label{fig:flower}
\end{figure}

\begin{proof}
Suppose $x$ is a sink for $j$ number of trees.
Let $G'$ be the graph induced by $x$ along with all the $j$ trees that form tree-sink structures with $x$ as the sink. Due to Lemma~\ref{lemma:induced}, there exists a local source and a local destination in $G'$. Note that if $x$ were either the local source or the local destination, then due to Corollary~\ref{corollary:rooted-tree-general} each of the trees requires a tracker in their vertex set, and hence the lemma holds. Suppose not. Let $Tr_i$ denote the $i^{th}$ tree and $G_i$ denote the graph induced by the vertex set of $Tr_i$ along with the vertex $x$, for $i\in[j]$. Then due to Lemma~\ref{lemma:induced}, each graph $G_i$ has at least one pair of local source and destination vertices. Consider induced graphs $G_1$ and $G_2$. See Figure~\ref{fig:flower}. Note that for $G_1$ there exists a path from $x$ to $t$ via $Tr_2$, that intersects with $G_1$ only at the sink $x$, thus making $x$ a local destination for $G_1$. Hence due to Corollary~\ref{corollary:rooted-tree-general}, at least one tracker is needed in $V(Tr_1)$. Next consider $G_2$. Note that there exists a path from $s$ to the sink $x$, via $Tr_1$, that intersects at $G_2$ only at $x$, thus making $x$ a local source for $G_2$. Hence by Corollary~\ref{corollary:rooted-tree-general}, at least one tracker is needed in $V(Tr_2)$. Since these arguments can be extended for any induced graph $G_i$, it holds that at least one tracker is required in the vertex set of each of the trees.
\end{proof}

Next we give two lemmas to bound the vertices in $V_1$ and edges in $E_1$.

\begin{lemma}
\label{lemma:bounding-V1-single-vertex}
For a vertex $f\in S$, the number of vertices in $\mathcal{F}$ that form tree-sink structures with $f$ as a sink is at most $3k$ in a YES instance.
\end{lemma}

\begin{proof}
Let $f\in S$ and $V_f\subseteq V(\mathcal{F})$ be the set of vertices that form tree-sink structures with $f$ as a sink. Let  $x$ be the number of trees that form tree-sink structures with $f$ as sink, each with $l_i$ number of vertices adjacent to $f$, where $i\in [x]$. Note that $|V_f|=\sum_{i=1}^x l_i$.

From Corollary~\ref{corollary:tree-generalized} it is known that if a tree-sink structure is formed such that the sink is adjacent to $\delta$ vertices of the tree, then at least $\delta-2$ trackers are required in the tree vertices. 
Hence, each of the trees forming tree-sink structures with $f$ as sink, require $l_i-2$ trackers in their vertex set, $i\in[x]$. Note that a tracker in one tree of a tree-sink structure cannot act as a tracker for a tree-sink structure with a disjoint tree. 
Since the total budget for trackers is $k$, $\sum_{i=1}^x (l_i -2) \leq k$.  From Lemma~\ref{lemma:flower}, it follows that $f$ can be a sink for at most $k$ tree-sink structures. Thus $x\leq k$. It follows that $|V_f|=\sum_{i=1}^x l_i \leq 3k$. 
\end{proof}

\begin{lemma}
\label{lemma:bounding-V1}
The number of vertices in $\mathcal{F}$ that share neighbors in $S$ with vertices from the same tree is at most $6k^2$ and the number of edges between these vertices and $S$ is at most $6k^2$ in a YES instance.
\end{lemma}

\begin{proof}
When two or more vertices from a tree in $\mathcal{F}$ share a common neighbor, say $f\in S$, they form a tree-sink structure, the tree being the minimal connected subtree containing all neighbors of $f$ in that tree, and $f$ being the sink.
Due to Lemma~\ref{lemma:bounding-V1-single-vertex} it is known that for a vertex $f\in S$, at most $3k$ vertices from $\mathcal{F}$ form tree-sink structures with $f$ as a sink.
Since $|S|\leq 2k$, the total number of vertices in $V_1$ is at most $(2k) 3k$ i.e. $6k^2$. As we considered only single edges between the sink and its neighbors in the trees (subgraphs in $\mathcal{F}$) of the tree-sink structures, $|E_1|\leq 6k^2$.
\end{proof}

\begin{Reduction Rule}
\label{red:bounding-V1}
If the number of vertices that share neighbors in $S$ with vertices from the same tree are more than $6k^2$, then we return a NO instance.
\end{Reduction Rule}

\begin{lemma}
\label{lemma:red-rule-V1}
Reduction Rule~\ref{red:bounding-V1} is safe and can be applied in polynomial time.
\end{lemma}
\begin{proof}
Safeness of the reduction rule follows from Lemma~\ref{lemma:bounding-V1}. To apply the rule, for each vertex $f$ in $S$, we consider the subgraph $G'$ induced by $G\setminus (S\setminus\{f\})$. There can be at most $n$ trees in $G'\setminus \{f\}$, and each tree can have at most $n$ vertices. For each tree we check if at least two vertices are adjacent to $f$. For tree that have at least two vertices adjacent to $f$, we count the number of such vertices. This can be done in $\Oh(n^2)$ time. Since $|S|\leq 2k$, the total time taken will be $\Oh(n^3)$.
\end{proof}

\subsection{Bounding $V_2$}
\label{subsec:V2}

Recall that
$V_2=\{ v\in\mathcal{F} \mid \exists u\in V(Tr), \textrm{ where } Tr \in \mathcal{F},  v\notin V(Tr),  N_S(u)\cap N_S(v)\neq\emptyset \}$.
Hence $V_2$ is the set of vertices that share a neighbor in $S$ with a vertex from another tree.
Note that the vertices in $V_2$ may or may not share a neighbor with a vertex in the same tree. As mentioned before we allow the possible over counting of vertices of $V_1$ here as the final bound calculated is still $\Oh(k^2)$. 

%

Observe that if a vertex $v\in V_2$ belongs to a tree $Tr\subseteq G(\mathcal{F})$ such that $|N(V(Tr))\cap S|=1$, then $v$ belongs to $V_1$ as well. Thus in such a case we need not count $v$ in $V_2$.
Excluding such vertices, we can assume that for each vertex $v\in V_2$ it holds that $|N(V(Tr))\cap S|\geq 2$, where $Tr\in\mathcal{F}$ is the tree to which $v$ belongs. This implies that either $v$ has at least two neighbors in $S$, or there exists a vertex in $V(Tr)\setminus\{v\}$ that is adjacent to a vertex in $S\setminus\{f\}$, where $f\in N(v)\cap S$ is adjacent to a vertex in another tree in $\mathcal{F}$. Since we need an upper bound on $V_2$, we assume the second case, i.e. for each vertex in $V_2$ there exists another vertex in the same tree and has a different neighbor in $S$.

Let $a,b,c\in V(Tr)$ where $Tr\subseteq G(\mathcal{F})$. If $u,v\in S$ and $u,v \in (N(a)\cup N(b) \cup N(c))$ then at least two vertices among $a,b,c$ share a neighbor in $S$ (either $u$ or $v$) and thus belong to $V_1$.
Hence we can assume that a pair of vertices in $S$ is adjacent to at most two vertices from $V_2$ from each tree in $\mathcal{F}$. 


\begin{lemma}
\label{lemma:no-tracker-path-vertices}
The number of vertices in $V(\mathcal{F})\cap N(S)$ that belong to vertex disjoint paths without any trackers, between a pair of vertices in $S$ are at most $6k(2k-1)$ in a YES instance.
\end{lemma}

\begin{proof}
Let $u,v\in S$ be a pair of vertices such that there exist three vertex disjoint paths (comprising of vertices from $\mathcal{F}$) between them. Let $G'$ be the subgraph induced by $u$ and $v$ along with the three vertex disjoint paths between them. If $u$ and $v$ are trackers, and are not a local source-destination pair for $G'$, then it is possible that no trackers are required on the three paths between them. See Figure~\ref{fig:no-trackers}. 

\begin{figure}[ht]
\centering
\includegraphics[scale=0.40]{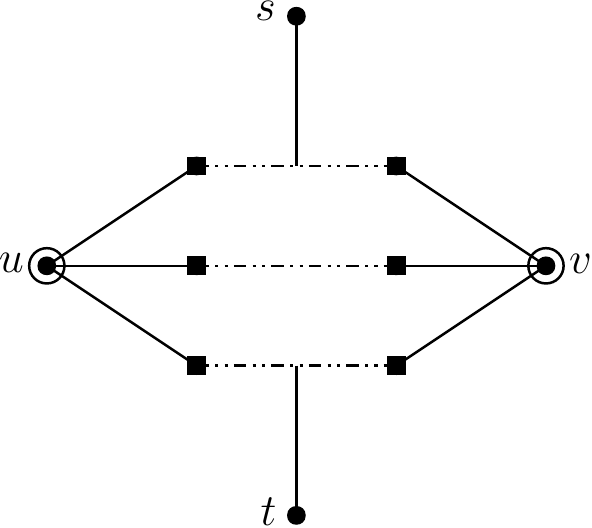} 
\caption{Three paths without any trackers between $u,v\in S$, when $u,v$ are trackers but do not form a local source-destination pair. Square vertices belong to $X\subseteq V_2'$.} 
\label{fig:no-trackers}
\end{figure}

Since $|S|\leq 2k$, there exist at most $2k\choose 2$ such pairs of $u,v\in S$. Further, for each of the three vertex disjoint paths between a pair $u,v\in S$, passing from vertices in $\mathcal{F}$ (with $u,v$ as end points), can account for at most two vertices from $V(\mathcal{F})\cap N(\{u,v\})$. Hence, the total number of vertices in $V(\mathcal{F})\cap N(S)$ that form vertex disjoint paths between vertices of $S$, and do not contain any trackers themselves are at most $6{2k\choose 2}$ i.e. $6k(2k-1)$.
\end{proof}

Next we give a lemma to bound the vertices in $V_2$.

\begin{lemma}
\label{lemma:bounding-V2}
The number of vertices in $\mathcal{F}$ that share a neighbor with a vertex from another tree is at most $28k^2-14k$ and the number of edges between these vertices and $S$ is at most $32k^2-16k$ in a YES instance.
\end{lemma}

%
\begin{figure*}[ht]
  \centering
\begin{minipage}[b]{0.49\textwidth}
\centering
\includegraphics[width=50mm]{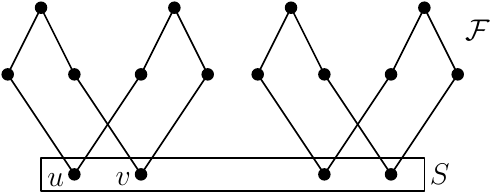}\\
{\small (a) Pairs of vertices in $\mathcal{F}$ share their neighbors with pair of vertices from another tree.}\\
\end{minipage}
 \begin{minipage}[b]{0.49\textwidth}
  \centering
 \includegraphics[width=40mm]{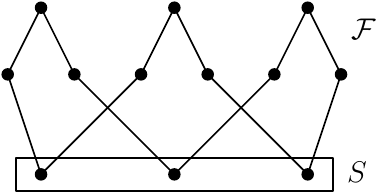}\\
{\small (b) No pair of vertices in $\mathcal{F}$ share their neighbors with pair of vertices from another tree.}\\
\end{minipage}
 \caption{Vertices in $\mathcal{F}$ sharing neighbors with vertices from other trees in $\mathcal{F}$.}
 \label{fig:sharing}
 \end{figure*}

\begin{proof}
We subdivide $V_2$ into following two sub-categories:
\begin{itemize}

\item Let $V_2'$ be the set of vertices such that each pair of vertices from a tree share their neighbors in $S$ with a pair of vertices from another tree. See Figure~\ref{fig:sharing}(a). 
Consider a pair of vertices $u,v\in S$. Observe that if pairs of vertices from different trees are incident to $u$ and $v$, they form vertex disjoint paths between $u$ and $v$ passing through the trees to which they belong.
Let $X$ be the set of vertices that form vertex disjoint paths without any trackers, between pairs of vertices of $S$. From Lemma~\ref{lemma:no-tracker-path-vertices}, $|X|\leq 6k(2k-1)$.

Next we consider those pairs of vertices in $S$ that are adjacent to at least four pairs of vertices of $V_2'$ leading to formation of at least four vertex disjoint paths. Observe that for such a pair of vertices in $S$, at least one tracker is needed on the paths passing through vertices of $V_2'$ and the pair. 
Let there be $m$ such pairs of vertices in $S$, and $Y\subseteq V_2'$ be the set of vertices considered in this case.
Let $P_i$ denote the pairs of vertices from $Y$ that are adjacent to the $i^{th}$ pair of vertices in $S$. 

From Lemma~\ref{lemma:disjoint} it is known that if there are more than $k+1$ number of vertex disjoints paths between a pair of vertices, then $G$ cannot be tracked with at most $k$ trackers. 
Then the total number of trackers required is $\sum_{i=1}^m (|P_i|-3)$. Since the total budget for the number of trackers is at most $k$, it follows that $\sum_{i=1}^m (|P_i|-3) \leq k$. Since each pair of vertices from $S$ considered here, requires at least one tracker on a vertex in $\mathcal{F}$ (as there are at least $4$ vertex disjoint paths between each pair), $m\leq k$. Hence $\sum_{i=1}^m |P_i| \leq 4k$, and hence the count of vertices forming the paths between the $m$ pairs, is less than or equal to $8k$. So far we assumed the vertex disjoint paths formed by vertices of $Y$ to be disjoint since we have counted trackers for each pair separately (and added them). However note that, trackers may be shared among different paths formed by vertices of $Y$. There might be two (or more) pairs of vertices in $S$ that share a neighbor $y\in Y$ (while forming vertex disjoint paths through other vertices in $\mathcal{F}$) and $y$ may be a tracker. 

Since $|S|\leq 2k$, at most $2k$ vertices from $S$ may be adjacent to $2k$ vertices of $Y$ from a single tree, with only one of these vertices being a tracker. Thus for each of the $8k$ vertices counted so far in $Y$, we may have at most $(2k-1)$ vertices of $Y$ that do not have a tracker. Hence $|Y|\leq 8k(2k-1)$.

Thus $|V_2'|= |X|+|Y| \leq 6k(2k-1)+8k(2k-1)=14k(2k-1)$.

\item Let $V_2''$ be the set of vertices that share a neighbor in $S$ with a vertex of another tree, but there does not exist a pair of vertices in any tree that shares their neighbors in $S$ with a pair of vertices from another tree, i.e. $V_2''=V_2\setminus V_2'$. See Figure~\ref{fig:sharing}(b). 
Observe that in this case, a pair of vertices, say $u,v\in S$, can be adjacent to a pair of vertices of $V_2$ from at most one tree. 
Since there are at most $2k \choose 2$ pairs in $S$, and each pair is adjacent to two vertices from $V_2''$, it holds that $V_2''\leq 2 {2k\choose 2} $. Hence $|V_2''|\leq 2k(2k-1)$. 
\end{itemize}
Hence $|V_2|=|V_2'|+|V_2''| \leq 16(2k-1)$ i.e. $32k^2-16k$. Since we have considered only single adjacencies between the vertices of $V_2'$ and $S$, it holds that $|E_2|\leq 32k^2-16k$.
\end{proof} 

\begin{Reduction Rule}
\label{red:bounding-V2}
If the number of vertices that share a neighbor in $S$ with vertices from same tree is more than $32k^2-16k$, then we return a NO instance.
\end{Reduction Rule}

\begin{lemma}
\label{lemma:red-rule-V2}
Reduction Rule~\ref{red:bounding-V2} is safe and can be applied in polynomial time.
\end{lemma}
\begin{proof}
Safeness of the reduction rule follows from Lemma~\ref{lemma:bounding-V2}. In order to implement the rule, first we go through all vertices in $\mathcal{F}$ and create a data structure that maintains which vertex belongs to which tree. This operation can be done in $\Oh(n^2)$ time. Next for each vertex $v\in\mathcal{F}$, if $Tr\in\mathcal{F}$ is the tree to which $v$ belongs, we check if a vertex in $N(v)\cap S$ has a neighbor in $V(\mathcal{F})\setminus V(Tr)$. Since the number of trees in $\mathcal{F}$ is $\Oh(n)$, we can do the check for all vertices of $\mathcal{F}$ in $\Oh(n^2)$ time.
\end{proof}

\subsection{Bounding $V_3$}
\label{subsec:V3}
Recall that $V_3=\{ v\in\mathcal{F} \mid \nexists u\in \mathcal{F} \textrm{ and } N_S(u)\cap N_S(v)=\emptyset \}$.
Hence $V_3$ is the set of vertices that do not share their neighbors in $S$ with other vertices in $\mathcal{F}$ i.e. the vertices in $V_3$ have an exclusive neighbor each in $S$. Since $|S|\leq 2k$, and each vertex in $S$ can be adjacent to only vertex of $V_3$, it holds that $|V_3|\leq 2k$. Thus we have the following lemma.

\begin{lemma}
\label{lemma:bounding-V3}
The number of vertices in $\mathcal{F}$ that do not share their neighbors in $S$ with other vertices in $\mathcal{F}$ is at most $2k$ in a YES instance.
\end{lemma}

\subsection{Bounding $V_4$}
\label{subsec:V4}

Recall that $V_4=\{v\in\mathcal{F} \mid N(v)\cap S=\emptyset\}$. Hence $V_4$ is the set of vertices in $\mathcal{F}$ that do not have any neighbors in $S$.
Note that all leaves in $\mathcal{F}$ necessarily have a neighbor in $S$ due to Reduction Rules~\ref{red:stpath} and \ref{red:no-deg-one}. Hence the vertices in $V_4$ are only the internal vertices of the trees in $\mathcal{F}$.

\begin{lemma}
\label{lemma:bounding-V4}
The number of vertices in $\mathcal{F}$ that do not have neighbors in $S$ is at most $78k^2- 15k$ in a YES instance.
\end{lemma}

\begin{proof}
Due to structural properties of a tree it is known that the number of non-leaf vertices in a tree with degree greater than $2$ is upper bounded by the number of leaves in the tree. Further due to Reduction Rule~\ref{red:twoinseries} between every closest pair of vertices with degree greater than $2$, in a tree, there can exist at most one vertex with degree $2$. Hence $|V_4|\leq 3(|V_1|+|V_2|+|V_3|)$. From Reduction Rules~\ref{red:bounding-V1} and \ref{red:bounding-V2}, it follows $|V_1|+|V_2|+|V_3|\leq 6k^2+32k^2-16k +2k$ i.e. $38k^2-14k$. Thus $|V_4|\leq 3(38k^2-14k)=114k^2-42k$.
\end{proof}

\subsection{Final Kernel}

\begin{theorem}
\label{thm:tp-kernel}
\tp admits a kernel of size $\Oh(k^2)$ in general graphs. 
\end{theorem}

\begin{proof}
The total number of vertices in $\mathcal{F}$ is less than or equal to $|V_1|+|V_2|+|V_3|+|V_4|$. Due to Reduction Rules~\ref{red:bounding-V1}, \ref{red:bounding-V2}, and Lemma~\ref{lemma:bounding-V4}, $|\mathcal{F}|\leq 6k^2+20k^2-7k+2k+114k^2-42k$. Hence $\mathcal{F}\leq 140k^2-47k$. Thus $|V(G)|\leq 140k^2-45k$ after including the vertices from $S$. From Lemmas~\ref{lemma:bounding-V1} and \ref{lemma:bounding-V2} it is known that the total number of edges between $\mathcal{F}$ and $S$ is at most $38k^2-16k$. Since $\mathcal{F}$ is a forest, the total number of edges in $\mathcal{F}$ will be upper bounded by $|\mathcal{F}|$, i.e. $140k^2-47k$. The maximum number of edges whose both end points lie in $S$ is $2k\choose 2$. Hence the total number of edges in $G$ is at most $38k^2-16k+140k^2-47k+2k^2-2k$, i.e. $180k^2-65k$. Since both the number of vertices and edges in $G$ is upper bounded by $\Oh(k^2)$, it holds that \tp admits a kernel of size $\Oh(k^2)$.
\end{proof}

\section{Linear Kernel for Planar Graphs}
\label{sec:planar}

In this section we show that an instance $(G,k)$ of \tp, where $G$ is a planar graph, can be reduced to an equivalent instance $(G',k')$ such that if $(G,k)$ is an YES instance then $|V(G')|=\Oh(k)$, $|E(G')|=\Oh(k)$ and $k'\leq k$. A linear kernel for planar graphs can be derived from an observation given by Eppstein et al.~\cite{ep-planar} as explained in the following theorem.

We start by applying Reduction Rules~\ref{red:stpath}, \ref{red:no-deg-one}, \ref{red:twoinseries}, \ref{red:deg-2-triangle}, and \ref{red:disjoint-m-paths}. Next we give a reduction rule that bounds the number of faces/regions in a reduced graph.

\begin{Reduction Rule}
\label{red:bound-faces}
In a reduced planar graph $G$, if the number of faces $|F|> 2k+1$, then $G$ does not contain a tracking set of size $k$.
\end{Reduction Rule}

\begin{lemma}
\label{lemma:red-bound-faces}
Reduction Rule~\ref{red:bound-faces} is safe and can be applied in polynomial time.
\end{lemma}

\begin{proof}
Let $G=(V,E)$ be a reduced planar graph with $F$ as the set of faces in $G$.
It is known from~\cite{ep-planar} that the number of faces $|F|\leq 2. OPT + 1$, where $OPT$ is the number of trackers in an optimum tracking set. If $OPT\leq k$, then $|F|\leq 2k+1$. Hence, if $|F|>2k+1$, then it is a NO instance. This proves the safeness of the reduction rule. We can calculate the value of $|F|$ using the Euler's formula: $|F|=|E|-|V|+2$, and thus the rule is applicable in polynomial time.
\end{proof}

\begin{theorem}
\label{thm:planar-kernel}
\tp admits a kernel of size $\Oh(k)$ in planar graphs. 
\end{theorem}

\begin{proof}
Let $G=(V,E)$ be a reduced planar graph with $F$ as the set of faces/regions in $G$. 
Let $V_{\geq 3}$ be the set of vertices with degree greater than or equal to $3$ and $V_2$ be the set of vertices with degree equal to $2$. 
Observe that after the application of Reduction Rules~\ref{red:stpath} and \ref{red:no-deg-one} there does not exist a vertex of degree one in the graph. Further, due to Reduction Rules~\ref{red:twoinseries} and  \ref{red:disjoint-m-paths} each vertex of degree two has vertices with degree three of more as its neighbors.
First we construct a graph $G'=(V',E')$ by short-circuiting (the vertex is deleted and an edge is introduced between its neighbors) all vertices in $V_2$. Due to Reduction Rules~\ref{red:deg-2-triangle},\ref{red:disjoint-m-paths} the short-circuiting does not create parallel edges. Note that the number of faces $|F|$ in $G$ is same as that in $G'$. Further, $|V'|=|V_{\geq 3}|$, $|V_2|\leq |E'|$ and $|E|\leq 2|E'|$.
We now try to bound the size of $V_{\geq 3}$ in $G'$.
Since summation of degrees of vertices in a graph is twice the number of edges, and degree of all vertices in $G'$ is at least three,
\begin{align}
\label{eq:1}
|E|\geq 3|V'|/2
\end{align}
Due to Euler's theorem,
\begin{align*}
|F| &=|E'|-|V'|+2 \\
&\geq |V'|/2 + 2 \textrm{  \;\;\;\; (Due to Equation~\ref{eq:1}) } \\
\textrm{Hence, }|V_{\geq 3}|=|V'| &\leq 2(|F|-2) \\
&\leq 2(2k+1 - 2) \textrm{  \;\;\;\; (Due to Reduction Rule~\ref{red:bound-faces}) } \\
&\leq 4k - 2
\end{align*}

Thus $|V_2|=|E'|=|F|+|V'|-2 \leq 2k+1 + 4k-2-2=6k-3$. The total number of vertices in $G$ is $|V|=|V_2|+|V_{\geq 3}|\leq 6k-3+4k-2 = 10k-5$. Since $|E|\leq 2|E'|$,  $|E|\leq 12k-6$. Hence, we have a kernel with $\mathcal{O}(k)$ vertices and $\mathcal{O}(k)$ edges.
\end{proof} 

\section{Hardness Result}
\label{sec:hardness}

Here we show that finding a tracking set of size at most $n-k$ for a graph $G$ with $n$ vertices is {\sc W[1]}-hard. 

\begin{theorem}
\label{thm:tp-hard}
For general graphs the problem of finding a tracking set of size at most $n-k$ in a graph of $n$ vertices, is {\sc W[1]}-hard.
\end{theorem}
\begin{proof}
\tp has been shown \NP-hard by a reduction from \vc in~\cite{tr-j}. Specifically it has been shown that given a graph $G=(V,E)$ on $n$ vertices one can construct in polynomial time a graph $G'(V',E')$, where $|V'|=n'=|V|+|E|+5$, such that $G$ has a vertex cover of size $k$ if and only if $G'$ has a tracking set of size $k+|E|+2$. 
It follows that $G$ has an independent set of size $k$ if and only if $G'$ has a tracking set of size $n-k+|E|+2$, i.e. $n'-k-3$.
Hence $G$ has an independent set of size $k+3$ if and only if $G'$ has a tracking set of size $n'-k$.
Since {\sc Independent Set} is {\sc W[1]}-hard~\cite{DF99}, it follows that the problem of finding a tracking set of size at most $n-k$ is {\sc W[1]}-hard as well.
\end{proof}

\section{Conclusions}
\label{sec:concl}

In this paper we give improved kernels for the \tp problems. This is achieved via exploiting the connection between a feedback vertex set and tracking set, structural properties of tress, and some counting arguments. 
Key open problems are finding an $\Oh(c^k)$ \FPT algorithm and kernel lower bounds for the problem. Other directions to explore are approximation algorithms and studying the problem for other graph classes like directed graphs and weighted graphs. 

%
%
%
\bibliographystyle{splncs04}
\bibliography{tracking}

\begin{thebibliography}{10}

\bibitem{tracking-binary}
Javed Aslam, Zack Butler, Florin Constantin, Valentino Crespi, George Cybenko,
  and Daniela Rus.
\newblock Tracking a moving object with a binary sensor network.
\newblock In {\em Proceedings of the 1st International Conference on Embedded
  Networked Sensor Systems}, SenSys ’03, page 150–161, New York, NY, USA,
  2003. Association for Computing Machinery.

\bibitem{twofvs}
V.~Bafna, P.~Berman, and T.~Fujito.
\newblock A 2-approximation algorithm for the undirected feedback vertex set
  problem.
\newblock {\em {SIAM} J. Discrete Math.}, 12(3):289--297, 1999.

\bibitem{ciac17}
A.~Banik, M.~J. Katz, E.~Packer, and M.~Simakov.
\newblock Tracking paths.
\newblock In {\em Algorithms and Complexity - 10th International Conference,
  CIAC 2017}, pages 67--79, 2017.

\bibitem{caldam18}
Aritra Banik and Pratibha Choudhary.
\newblock Fixed-parameter tractable algorithms for tracking set problems.
\newblock In {\em Algorithms and Discrete Applied Mathematics - 4th
  International Conference, {CALDAM} 2018, Guwahati, India, February 15-17,
  2018, Proceedings}, pages 93--104, 2018.

\bibitem{tr-j}
Aritra Banik, Pratibha Choudhary, Daniel Lokshtanov, Venkatesh Raman, and Saket
  Saurabh.
\newblock A polynomial sized kernel for tracking paths problem.
\newblock {\em Algorithmica}, 82(1):41--63, 2020.

\bibitem{tr1-j}
Aritra Banik, Pratibha Choudhary, Venkatesh Raman, and Saket Saurabh.
\newblock Fixed-parameter tractable algorithms for tracking shortest paths.
\newblock {\em CoRR}, abs/2001.08977, 2020.
\newblock URL: \url{http://arxiv.org/abs/2001.08977}, \href
  {http://arxiv.org/abs/2001.08977} {\path{arXiv:2001.08977}}.

\bibitem{seymour-chapter}
Daniel Bienstock and Michael~A. Langston.
\newblock Chapter 8 algorithmic implications of the graph minor theorem.
\newblock In {\em Network Models}, volume~7 of {\em Handbooks in Operations
  Research and Management Science}, pages 481 -- 502. Elsevier, 1995.

\bibitem{tracking-moving-obj2}
{Chih-Yu Lin}, {Wen-Chih Peng}, and {Yu-Chee Tseng}.
\newblock Efficient in-network moving object tracking in wireless sensor
  networks.
\newblock {\em IEEE Transactions on Mobile Computing}, 5(8):1044--1056, 2006.

\bibitem{iwoca}
Pratibha Choudhary.
\newblock Polynomial time algorithms for tracking path problems.
\newblock In {\em Combinatorial Algorithms - 31st International Workshop,
  {IWOCA} 2020, Bordeaux, France, June 8-10, 2020, Proceedings}, pages
  166--179, 2020.

\bibitem{polytime-arxiv}
Pratibha Choudhary.
\newblock Polynomial time algorithms for tracking path problems.
\newblock {\em CoRR}, abs/2002.07799, 2020.
\newblock URL: \url{https://arxiv.org/abs/2002.07799}, \href
  {http://arxiv.org/abs/2002.07799} {\path{arXiv:2002.07799}}.

\bibitem{tracking-moving-obj}
Alminas {\v{C}}ivilis, Christian~S. Jensen, and Stardas Pakalnis.
\newblock {\em Tracking of Moving Objects with Accuracy Guarantees}, pages
  285--309.
\newblock Springer Berlin Heidelberg, Berlin, Heidelberg, 2007.

\bibitem{Cygan:2015:PA:2815661}
M.~Cygan, F.~V. Fomin, L.~Kowalik, D.~Lokshtanov, D.~Marx, M.~Pilipczuk,
  M.~Pilipczuk, and S.~Saurabh.
\newblock {\em Parameterized Algorithms}.
\newblock Springer Publishing Company, Incorporated, 1st edition, 2015.

\bibitem{DF99}
R.~G. Downey and M.~R. Fellows.
\newblock {\em Parameterized Complexity}.
\newblock Springer-Verlag, 1999.

\bibitem{ep-planar}
David Eppstein, Michael~T. Goodrich, James~A. Liu, and Pedro Matias.
\newblock Tracking paths in planar graphs.
\newblock In {\em 30th International Symposium on Algorithms and Computation,
  {ISAAC} 2019, December 8-11, 2019, Shanghai University of Finance and
  Economics, Shanghai, China}, pages 54:1--54:17, 2019.

\bibitem{identifying-path}
Florent Foucaud and Matjaz Kovse.
\newblock On graph identification problems and the special case of identifying
  vertices using paths.
\newblock In S.~Arumugam and W.~F. Smyth, editors, {\em Combinatorial
  Algorithms, 23rd International Workshop, {IWOCA} 2012, Tamil Nadu, India,
  July 19-21, 2012, Revised Selected Papers}, volume 7643 of {\em Lecture Notes
  in Computer Science}, pages 32--45. Springer, 2012.

\bibitem{pathcover}
Florent Foucaud and Matjaz Kovse.
\newblock Identifying path covers in graphs.
\newblock {\em J. Discrete Algorithms}, 23:21--34, 2013.

\bibitem{seymour2}
Neil Robertson and P~D Seymour.
\newblock Graph minors. v. excluding a planar graph.
\newblock {\em J. Comb. Theory Ser. B}, 41(1):92–114, August 1986.

\bibitem{seymour1}
Neil Robertson and P.~D. Seymour.
\newblock Graph minors. xiii: The disjoint paths problem.
\newblock {\em J. Comb. Theory Ser. B}, 63(1):65–110, January 1995.

\bibitem{tracking-neural}
G.~Schram, F.X. [van~der Linden], B.J.A. Kröse, and F.C.A. Groen.
\newblock Visual tracking of moving objects using a neural network controller.
\newblock {\em Robotics and Autonomous Systems}, 18(3):293 -- 299, 1996.
\newblock IAS-4 Conference.

\bibitem{tracking-framework}
Bi~Song, Ting-Yueh Jeng, Elliot Staudt, and Amit~K. Roy-Chowdhury.
\newblock A stochastic graph evolution framework for robust multi-target
  tracking.
\newblock In {\em Computer Vision -- ECCV 2010}, pages 605--619. Springer
  Berlin Heidelberg, 2010.

\bibitem{coordinated-tracking}
Q.~{Zhang}, L.~{Lapierre}, and X.~{Xiang}.
\newblock Distributed control of coordinated path tracking for networked
  nonholonomic mobile vehicles.
\newblock {\em IEEE Transactions on Industrial Informatics}, 9(1):472--484,
  2013.

\end{thebibliography}
%
%
%
%
%


\end{document}